\documentclass[article, 12pt]{article}

\usepackage[T1,T2A]{fontenc}
\usepackage[utf8]{inputenx}
\usepackage[russian,english]{babel}
\usepackage[numbers]{natbib}
\usepackage{amsfonts}
\usepackage{amsmath}
\usepackage{amsthm}
\usepackage{amssymb}
\usepackage{footmisc}
\usepackage{graphicx} 
\usepackage{pdfpages}
\usepackage{epstopdf}

\usepackage{setspace}
\theoremstyle{plain}
\newtheorem{thm}{Theorem}
\newtheorem{defin}{Definition}

\newtheorem{corol}{Corollary}

\newtheorem{lemma}{Lemma}

\begin{document}


\title{Uncertainty and absence of arbitrage opportunity}%


\author {Y. Ivanenko \footnote{IAE de Paris, France, yarsky1@gmail.com}, \hspace{10pt}{I. Pasichnichenko} \footnote{Institute for Applied System Analysis, Kyiv, Ukraine, io.pasich@gmail.com}}

\date{\vspace{-5ex}}

\maketitle

\begin{abstract}
It is shown that absence of arbitrage opportunity in  financial markets is a particular case of existence of uncertainty in decision system. Absence of arbitrage opportunity is considered in the sense of the Arrow-Debreu model of financial market with a riskless asset, while uncertainty (or ambiguity) is defined on the basis of the principle of internal coherence of M. Allais.
\end{abstract}

\section{Introduction}

The notion of uncertainty as of something opposite to risk was proposed by Frank Knight \cite{K}. The presence or availability of probability distribution on the set of consequences or on the set of states of Nature was considered as \textit{risk}, while the absence or unavailability of such probability distribution was considered as \textit{uncertainty}. This paradigm considerably influenced the consequent development of the so called \textit{descriptive} decision theory \cite{S, Bell}. So decision theories, explaining the Ellsberg paradox, essentially recurred to the definition of uncertainty proposed by Knight and often used the words \textit{ambiguity} and \textit{uncertainty} as synonyms \cite{E, Bouy}.

There exist, however, several arguments in favour of an alternative view on the problem of formalisation of the notion of uncertainty or ambiguity. One of such arguments is the possibility to define these notions solely on the basis of what Maurice Allais called "the principle of internal coherence employed in social sciences"\footnote{«The principle of internal consistency implies only: (a) the use of objective probabilities when they exist, and (b) the axiom of absolute preference which states that out of two situations, one is certainly preferable if, for all possible outcomes, it yields a greater gain.»(\cite{A}, p.504).}. Assuming, following Allais, that rationality implies that decisions are made for the sake of achieving goals (see \cite{A}, p.504), one can define uncertainty as existence of several preference relations on the set of decisions (acts) when the preference relation on consequences is fixed \cite{IM2008, I2010}. \footnote{In other words, the person knows what does she want, but does not know which way to choose in order to achieve that.\label{more}} It can be shown, that existence of at least one pair of acts on which the preference relation is not unique is already necessary and sufficient for existence of uncertainty.

It turns out that the stronger requirement, that uncertainty, thus defined, exists for any pair of acts is analogical to the notion of absence of arbitrage opportunity used in mathematical finance \cite{V, AV, D}. This means that the notion of absence of arbitrage opportunity is the particular case of existence of uncertainty. Therefore one can consider many subjects that traditionally have been considered as specific exclusively to financial theory, in general, and mathematical finance, in particular, from the vantage point of decision theory, and remaining strictly within its borders. \footnote{Note that the notion of arbitrage decision in game theory is opposite to the one employed in finance. Recall that in game theory a decision is called the arbitrage decision if it satisfies all players \cite{N, LR}, while in financial theory, when speaking about arbitrage opportunity, it is always understood as satisfaction of interests of only one person.}

Note, that decision theory based on this definition of uncertainty leads to the following distinction between risk and uncertainty. Risk becomes the property of one decision, meaning solely the multiplicity of consequences of the decision, without any relation to whether or not there is a probability distribution. Uncertainty (or ambiguity) is the property of decision situation, that is when there are many decisions, each of which may have many consequences \cite{I2010}. That is, uncertainty can exist there where previously (according to Knight) there could be only risk. Moreover, in this formulation the notion of uncertainty (or ambiguity) makes sense only relatively to the goals of decision maker, described by her preference relation on the set of consequences. \footnote{Remark, that this approach to risk and uncertainty is in line with the notion of  \textit{coherent} uncertainty measures \cite{deGr, Mar, I2010}.} That is why it makes sense to consider the approach to formalisation of the notion of uncertainty (or ambiguity) proposed here as \textit{normative} \cite{Bell} or \textit{operational}.

The article is organised in the following way. In Section 2 the definition and the criterion of existence of uncertainty in decision situation are introduced. Note that the notions of uncertainty and ambiguity are understood as synonyms. In Section 3 it is shown that, on the one hand, the absence of arbitrage opportunity on financial markets is a particular case of existence of uncertainty in decision situation, and, on the other hand, it is possible to extend the notion of arbitrage opportunity from financial markets to the general decision situation. Thus Section 3 can be considered as a concrete, though theoretical, example of the abstract framework of Section 2.

\section{Uncertainty}

Pursuing the examination of conditions of existence of uncertainty, initiated in \cite{IM2008, I2010}, consider the problem of formal description of the class of decision making situations, where the preference of the individuum relatively to consequences does not "completely define" her preferences relatively to decisions (acts). In other words, consider such situations where preference relation on consequences is compatible with several preference relations on acts at the same time. We need to refine the intuitive understanding of how preferences on consequences can define preferences on acts. This refinement may have the form of the condition imposed on the relation between the two preference relations. This condition should be weak enough in order to span the largest class of possible behavior. This article employs the condition of domination, which is necessary in order to formalize mathematically the principle of internal consistency mentioned above.

Due to terminological discrepancies in contemporary literature, let us recall here properties of the binary relations that are used in this article. Let $X$ be an arbitrary set.
\begin{defin}\label{def1}
Relation $(\prec, X)$  is called preference relation if it is:
\begin{itemize}
\item[1)] assymetric: $x \prec y \Rightarrow \text{ not } y \prec x$ for all $x,y \in X$;
\item[2)] negatively transitive: $\text{ not } x \prec y, \text{ not } y \prec z \Rightarrow \text{ not } x \prec z $ for all $x,y,z \in X$.
\end{itemize}
\end{defin}
Thus the preference relation is the irreflexive version of a common weak order. Having defined the preference relation one can define the indifference relation in the usual way, assuming
\begin{equation*}
x \sim y \Leftrightarrow (\text{ not } x \prec y \text{ and } \text{ not } y \prec x ) \text{ for all } x,y \in X
\end{equation*}
That indifference relation is equivalence is straightforward.
\begin{defin}\label{def2}
Relation $(\prec, X)$  is called strict partial order if it is:
\begin{itemize}
\item[1)] irreflexive: $\text{ not } x \prec x$ for all $x \in X$;
\item[2)] transitive: $x \prec y, y \prec z \Rightarrow x \prec z$ for all $x,y,z \in X$.
\end{itemize}
If the strict partial order is, besides,
\begin{itemize}
\item[3)] connected: $x \neq y \Rightarrow x \prec y \text{ or } y \prec x$ for all $x,y \in X$,
\end{itemize}
then it is called strict linear order. \footnote{Strict linear order is a preference relation as well.}
\end{defin}

Preference relation on the set of consequences reflects their desirability to the decision maker or, in other words, her interests. Consider, for simplicity, as consequences real numbers with their natural order $<$ as the preference relation (see footnote \ref{more}). Here, for real numbers, the only meaningful operation is their comparison in value. This identification of consequences with real numbers fits all the cases where the conditions of the theorem of existence of utility function are satisfied (see, for instance, \cite{F}, p.27). It is clear, however, that one and the same real life situation with consequences of arbitrary nature may have for different persons the form of different situations with numeric consequences, depending on personal preferences of consequences.

\begin{defin}\label{def3}
A pair $(\Theta, D)$, where $D$ is some set of real functions on the set $\Theta$, is called a matrix decision scheme (matrix scheme), $D$ is called the set of acts (decisions), $\Theta$ is the set of values of some  parameter $\theta$.
\end{defin}

The set of acts contains the acts available in the decision situation. The consequences of each of these acts are, in the general case, not unique, which is described by means of dependence of consequences, i.e. values of functions from the set $D$, on the parameter $\theta \in \Theta$, usually called the state of Nature \cite{S}. The matrix scheme does not contain any complementary information about the regularity of appearance of consequences. For instance, it does not contain a probability distribution on the set $\Theta$. That is why such pairs of sets are called "schemes".\footnote{It is worth noticing that a matrix decision scheme thus defined is a combination of the usual Savage decision-theoretical structure, consisting of the set of acts, $A$, of the set of states of Nature, $\Theta$, of the set of consequences, $C$, of the mapping $g: A \times \Theta \rightarrow C$, and of the utility function on consequences, $u: C \rightarrow \mathbb{R}$. That is in terms of \cite{IM2008, I2010} we are dealing here with the pair $(Z_m, \beta_C), Z_m=(\Theta, A, C, g)$, where the preference relation of  consequences $\beta_C$ is represented by some utility function $u(g(\theta,a))$.}

\begin{defin}\label{def4}
A relation $(\prec, D)$ defined on the set of acts is called \textbf{domination} if for any $d_1, d_2 \in D$\\
$d_1 \prec d_2 \Leftrightarrow \left[ d_1(\theta) \leq d_2(\theta) \text{ for all } \theta \in \Theta \text{ and } d_1(\theta^*) < d_2(\theta^*) \text{ for some } \theta^* \in \Theta \right]$
\end{defin}

Apparently, domination is a strict partial order. We are interested here only in those preference relations on the set of acts that are \textit{compatible} with domination in the sense of the following definition.

\begin{defin}\label{def5}
The preference relation $(\prec^{P},D)$ defined on the set of acts of the matrix scheme $(\Theta,D)$ is called a projection of preference of consequences in the scheme $(\Theta,D)$ (or just a projection), if
\begin{equation}\label{eq1}
d_1 \prec d_2 \Rightarrow d_1 \prec^{P} d_2 \text{ for any } d_1, d_2 \in D,
\end{equation}
where $(\prec, D)$ is domination.
\end{defin}

\begin{defin}\label{def6}
A matrix scheme $(\Theta,D)$ contains uncertainty if the projection of preference of consequences is not unique.
\end{defin}
Thus, definition \ref{def6} divides matrix decision schemes and, correspondingly,  decision situations, on two classes: containing and not containing uncertainty.

The next theorem answers divers questions relatively to uniqueness of the projection and, correspondingly, the question of existence of uncertainty in decision scheme.
\begin{thm}\label{the1}
In any matrix scheme
\begin{itemize}
\item[1)] there exist a projection of preference of consequences;
\item[2)] projection of preference of consequences is unique if and only if domination is connected;
\item[3)] uniqueness of projection of preference of consequences implies its identity with domination.
\end{itemize}
\end{thm}
\begin{proof} Let $(\Theta, D)$ be an arbitrary matrix scheme.\\
1) Existence of projection follows imediatley from  Szpilrajn extension theorem (see, for instance, \cite{F}, p.31), according to which any strict partial order can be extended to strict linear order. Applying this theorem to domination $(\prec, D)$, we obtain a strict linear order $(\prec^{P}, D)$ such that
\begin{equation}\label{eq2}
d_1 \prec d_2 \Rightarrow d_1 \prec^{P} d_2 \text{ for any } d_1, d_2 \in D,
\end{equation}
guaranteeing largely fulfillement of the conditions of definition \ref{def5}.\footnote{A projection  should be a preference relation, not necessarily a strict linear order.}\\

2) \textbf{Sufficiency}. If domination is connected, then it is a preference relation and, hence, a projection of preference of consequences. Assume there is a projection $(\prec^{P},D)$ different from domination. That is assume there exist such $d_1, d_2 \in D$ that, $d_1 \prec^{P} d_2$ and not $d_1 \prec d_2$. But due to connectedness of domination, from $(\text{ not } d_1 \prec d_2)$ follows that either $d_1 = d_2$, or $d_2 \prec d_1$. In the second case, condition (\ref{eq1}) implies $d_2 \prec^{P} d_1$. So, in both cases assymetry of projection gives contradiction with the initial assumption that $d_1 \prec^{P} d_2$. \\

\textbf{Necessity}. Assume that domination is not connected and construct two different projections. Let none of the relations $d_1=d_2, d_1 \prec d_2, d_2 \prec d_1$ is fulfilled for some $d_1, d_2 \in D$. Construct the projetion $(\prec_{1}^{P}, D)$ as described above in section 1 of the proof. Such projection, moreover, will be a strict linear order. Therefore, it implies either $d_1 \prec_{1}^{P} d_2$ or $d_2 \prec_{1}^{P} d_1$. Let, for instance, the first case takes place. Construct now the projection $(\prec_{2}^{P}, D)$, for which $d_2 \prec_{2}^{P} d_1$ holds. For the sake of this construction extend the domination $(\prec, D)$ in a suitable way to $(\prec^{'}, D)$, assuming for any $x,y \in D$, that
\begin{equation}\label{eq3}
x \prec^{'} y \Leftrightarrow x \prec y \text{ or } \left[ (x \prec d_2 \text{ or } x = d_2) \text{ and } (d_1 \prec y \text{ or } y=d_1) \right].
\end{equation}
It is clear that $d_2 \prec^{'} d_1$ and
\begin{equation}\label{eq4}
x \prec y \Rightarrow x \prec' y \text{ for any } x,y \in D.
\end{equation}
It can be shown by sorting out all the possibilities in the right hand side of (\ref{eq3}), that $(\prec', D)$ is a strict partial order. Now repeat the reasonings of the section 1 of the proof for $(\prec', D)$ instead of $(\prec, D)$ and obtain the strict linear order $(\prec_{2}^{P})$ with $d_2 \prec_{2}^{P} d_1$. Condition (\ref{eq1}) is fulfilled due to (\ref{eq4}) and (\ref{eq2}). Hence we obtained the projection different from the previous one.

The statement 3 of the theorem follows immediately from the statement 2, since connectedness of domination turns this relation from strict partial order into strict linear order and, hence, into a projection of preference of consequences.
\end{proof}

In definition \ref{def6} we try to convey intuitive ideas about existence of uncertainty in matrix decision scheme.
At the same time, statement 2 of theorem \ref{the1} indicates how one can simplify the formal usage of this notion. This is done in the next proposition, providing the criterion of existence of uncertainty in matrix scheme.

\begin{corol}
Matrix scheme $(\Theta,D)$ contains uncertainty if and only if there are such $d_1, d_2 \in D$ and such $\theta, \theta' \in \Theta$, that $d_1(\theta) < d_2(\theta)$ and $d_2(\theta')<d_1(\theta')$.
\end{corol}
\begin{proof}
By definition, the matrix scheme contains uncertainty if and only if the projection of preference of consequences is not unique. According to statement 2 of theorem \ref{the1}, this is equivalent to non connectedness of domination. That is, there are such distinct $d_1, d_2 \in D$ that neither $d_2 \prec d_1$, nor $d_1 \prec d_2$ is true. This, in its own turn, by definition of domination, is equivalent to existence of such $\theta, \theta' \in \Theta$, that $d_1(\theta) < d_2(\theta)$ and $d_2(\theta')<d_1(\theta')$.
\end{proof}

Hence, even if a decision making situation contains only two states of Nature, there are two decisions, the consequences of which so depend on the states of Nature, that the matrix scheme contains uncertainty. The following figures 1 and 2 illustrate these results.
\begin{figure}[htb]
\begin{center}
\includegraphics[height=7cm,width=10cm]{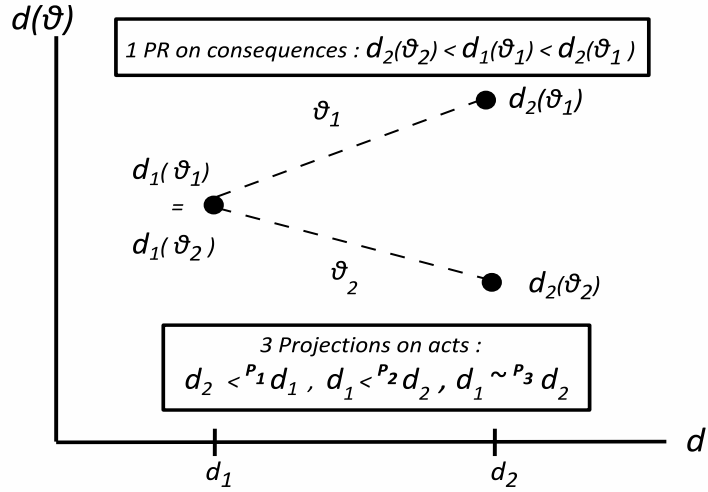}
\caption{An illustration of the situation with uncertainty.}\label{uncertainty}
\end{center}
\end{figure}

\begin{figure}[htb]
\begin{center}
\includegraphics[height=7cm,width=10cm]{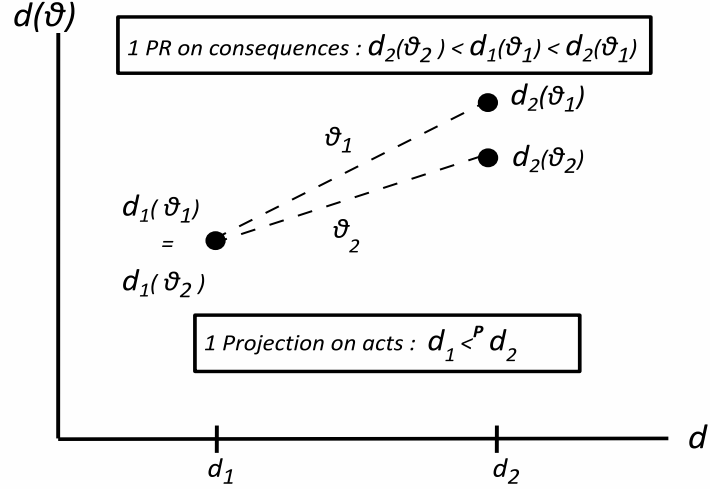}
\caption{An illustration of the situation without uncertainty.}\label{nouncertainty}
\end{center}
\end{figure}

Note that this approach to formalisation of the notion of uncertainty is in logical correspondence with Allais's principle of internal consistency, cited above. Namely, if projection is unique, then the act that dominates another act in all states is always preferable.

\section{Absence of arbitrage opportunity}

The notion of uncertainty in matrix scheme finds its immediate application in finance. It turns out that the notion of absence of arbitrage opportunity on the financial market can be considered as an example of existence of uncertainty in matrix decision scheme.

In the case of financial decision making situations, the convention related to numeric valuation of consequences does not play any role, since the initial situation contains only monetary consequences. We can assume that all decision makers prefer higher monetary sums to lower. By the same token, we assume that monetary sums play the role of the utility or, in other words, the values of the utility function on consequences coincide, in this case, with consequences themselves. This fact, in particular, will allow us considering financial decision schemes as examples of the matrix decision scheme presented in the previous section.

Recur to the following formal description of the financial market, related to the Arrow-Debreu model of securities market \cite{V, AV, D}.
\begin{defin}\label{def7}
A pair $(A,p), A \in \mathbb{R}^{n \times m}, p \in \mathbb{R}^{n}$ is called financial market with a riskless asset, if $A_{1j}=1$ for any $j \in \left\{1,\ldots,m\right\}$ and $p_1 \leq 1$.
\end{defin}
The elements of the set $\Theta=\left\{1,\ldots,m\right\}$ are called the states of the market. The matrix $A$ contains cash-flows of each of the $n$ securities for all states of the market. So that $A=\left[ a_1\hspace{5pt}a_2\ldots a_m \right]$, where $a_\theta=\left[a_{1\theta}\hspace{5pt}a_{2\theta} \ldots  a_{n\theta}\right]^T \in \mathbb{R}^n$, contains securities cash-flows in the state $\theta$. The vector $p \in \mathbb{R}^n$ consists of prices of the securities. Without loss of generality one can consider $p_1=1$. \footnote{Note that the riskless interest rate is equal to $\frac{1-p_1}{p_1}$. If $p_1 < 1$ then one can discount initial prices with the discount factor $1 + \frac{1-p_1}{p_1}= \frac{1}{p_1}$.} A portfolio is a vector $x \in \mathbb{R}^n$ and the components of its aggregate cash-flow in each state are the elements of  $A^{T}x \in \mathbb{R}$.

\begin{defin}\label{def8} Portfolio $x \in \mathbb{R}^n$ is called an arbitrage portfolio on the market $(A,p)$ if either $(p^Tx \leq 0 \text{ and } A^Tx >0)$ or $(p^Tx < 0 \text{ and } A^Tx \geq 0)$.
\end{defin}

Usually (see, for instance, \cite{AV,D}) this definition is followed by The Arbitrage Theorem, stating that the market does not contain an arbitrage portfolio if and only if there is a vector $\psi \in \mathbb{R}^n$ with positive components, called state prices, such that $p=A\psi$. However, this theorem is not the object of attention in this article, since our goal here is establishing of the relation between the notion of absence of arbitrage opportunity and  the notion of existence of uncertainty in matrix scheme.

Represent the financial market $(A,p)$ in terms of the matrix scheme $(\Theta, D)$, where $\Theta=\left\{ 1,2,\ldots,m\right\}$. Define the matrices $B=A- \left[p\hspace{5pt} p \ldots p\right], B \in \mathbb{R}^{n \times m}$. The set of acts $D=\left\{B^Tx: x \in \mathbb{R}^n \right\}, D \subset \mathbb{R}^m$ consists of all vectors of profits available on the market. The components of each such vector are profits of the portfolio $x$ in the corresponding state. Such matrix scheme $(\Theta, D)$ is called \textit{corresponding} to the market $(A,p)$. \footnote{In the terms of Savagian decision-theoretical structure we deal here with the following matrix scheme: the set of states of Nature, $\Theta=\left\{1,\ldots,m\right\}$, the set of acts $X=\mathbb{R}^n$, the set of consequences, $C=\mathbb{R}$, the mapping $g:\Theta \times X \rightarrow C$, such that $g(\theta,x)=\sum_{i=1}^{n}(a_{i\theta}-p_i)x_i, g(\cdot,\cdot) \in C$ and the utility function $u(g(\theta,x))=g(\theta,x), x \in X, \theta \in \Theta$.}

\begin{lemma}\label{lem1}
The market $(A,p)$ contains arbitrage portfolio if and only if the corresponding matrix scheme $(\Theta, D)$ contains a decision $d\in D$ such that $d >0$.\footnote{Here and in what follows we use the following conventional notation: for two vectors $x,y \in \mathbb{R}^n$, $x \leq y$ means $x_i \leq y_i, \forall i=1,\ldots,n$ and $x < y$ means $(x \leq y$  and $x \neq y)$..\label{vecs}}
\end{lemma}
\begin{proof}
From definition \ref{def8}, for an arbitrage portfolio $x$, in both cases, we have $a^T_\theta x - p^Tx \geq 0$ for any $\theta \in \Theta$ and $a^T_{\theta^{*}} x - p^Tx > 0$ for some $\theta^* \in \Theta$. Hence the act $d=B^Tx$ satisfies condition $d >0$. And vice versa, if there is $d >0, d \in D$, then for some portfolio $x$  we have $B^Tx=d>0$. Define $\bar{x} \in \mathbb{R}^n$ as $\bar{x}^T= x^T- \left[p^Tx \hspace{5pt}0\hspace{5pt}0\ldots\hspace{5pt}0 \right]$, then
\begin{equation}
A^T\bar{x}=A^T \left( x - \left[\begin{array}{c}p^Tx\\0\\ \vdots \\0  \end{array}\right]\right)=A^Tx - \left[\begin{array}{c}p^Tx\\p^Tx\\ \vdots \\p^Tx  \end{array}\right]=B^Tx>0,
\end{equation}
and $p^T\bar{x}=p^Tx - p^Tx=0$. Hence, portfolio $\bar{x}$ is an arbitrage portfolio according to the first condition of definition \ref{def8}. Note that portfolio $\bar{x}$ is just the leveraged version of the portfolio $x$. The lemma is proved.
\end{proof}
Hence the following theorem.
\begin{thm}\label{the2}
The market $(A,p)$ does not contain an arbitrage portfolio if and only if for the corresponding matrix scheme $(\Theta,D)$ and for any $D' \subseteq D$ such that $|D'|\geq 2$, the matrix scheme $(\Theta, D')$ contains uncertainty.
\end{thm}
\begin{proof}
For convenience, let us operate with negations of the  statements, logical equivalence of which is the object of the theorem. Note that for vectors, the relation of domination from definition \ref{def4} is identical to $(<,D)$ (see footnote \ref{vecs}). If market $(A,p)$ contains an arbitrage portfolio, then according to the previous lemma $0<d$ for some $d \in D$. Then $\left\{d,0\right\} \subseteq D$ and the matrix scheme $(\Theta,\left\{d,0\right\})$ does not contain uncertainty, since there is only one projection of preference of consequences $(\prec^{P},D)$. Namely, for any $x,y \in \left\{d,0\right\}$, $x \prec^P y$ if and only if $x=0, y=d$. Conversely, if for some $D' \subseteq D, |D'|\geq 2$, the matrix scheme $(\Theta, D')$ does not contain uncertainty, then, according to statement 2 of theorem 1, domination in this scheme is connected, and it is possible to chose different $d_1, d_2 \in D'$ such that, for instance, $d_1 < d_2$. If $d_1=B^Tx_1$ and $d_2=B^Tx_2$, where $x_1, x_2 \in \mathbb{R}^n$, then $\bar{d}=B^T(x_2 - x_1) \in D$ and $\bar{d}>0$, which, according to the lemma, implies existence of arbitrage portfolio.
\end{proof}

It is possible to relax the condition of absence of arbitrage opportunity on the financial market, considering portfolio $x$ as arbitrage portfolio only when $(p^Tx <0 \text{ and } A^Tx \geq 0)$. Such market sometimes is called weakly arbitrage-free \cite{D}. 
For such definition of the arbitrage portfolio the statement similar to theorem \ref{the2} is obtained if in the definition of projection of preference of consequences in matrix scheme the relation of domination is replaced by its enforced version: for any $d_1, d_2 \in D$
\begin{equation}
d_1 \prec^* d_2 \Leftrightarrow d_1(\theta) < d_2(\theta) \text{ for all } \theta \in \Theta.
\end{equation}

So theorem \ref{the2} states that absence of arbitrage opportunity in the financial market is equivalent to existence of uncertainty in any matrix scheme obtained by restriction of the decision set of the matrix scheme corresponding to that market. In other words, this is the case when uncertainty does not disappear from the matrix scheme in the result of deleting of or forbidding some acts. If the market contains arbitrage portfolio, then there is a set of acts (portfolios) that forms the matrix scheme which does not contain uncertainty.

Thus, theorem \ref{the2} allows considering one of the fundamental notions of financial theory without leaving the precincts of decision theory and, at the same time, without recurring to the notion of risk attitude of decision maker. This is all the more significant since it renders financial theory a particular case or, better, an example of decision theory, the idea usually disregarded. \footnote{For instance, in \cite{Back} one reads: "The option pricing formula of Black and Scholes (1973) is valid without regard to the risk preferences of investors, because it is based solely on the absence of arbitrage opportunities."}.

On the other hand, taking definition \ref{def6} and theorem \ref{the2}, one can extend the notion of arbitrage opportunity from financial markets to any decision making situation. Namely, one can say that any matrix scheme, that does not satisfy the conditions of theorem \ref{the2}, contains an arbitrage opportunity. In other words, the notion of arbitrage acquires the meaning of the univalent projecting of the preference of consequences onto the preference of acts. The opportunity of arbitrage means the possibility of such univalent projecting for some pair of acts in the scheme.

\section{Conclusions}

Thus, this paper gives definition of the notion of uncertainty in matrix decision scheme and provides the criterion of its existence. This formalisation of the notion of uncertainty is a farther development of the corresponding definition and theorem, proposed in \cite{IM2008, I2010} for the lottery decision scheme. Together these results comprise an alternative to Knightian approach to the problem of formalisation of the notion of uncertainty. At the same time, this alternative formulation corresponds to the principle of internal consistency of M. Allais.

It turns out that this approach to the notion of uncertainty in decision scheme is related to the notion of arbitrage opportunity on financial markets. It is shown that, on the one hand, absence of arbitrage opportunity (i.e. absence of arbitrage portfolio) on the financial market is a particular case of existence of uncertainty in matrix scheme, and, on the other hand, it is possible to extend the notion of arbitrage opportunity on any, not necessarily financial, decision making situation.

\section{Aknowlegements}

The authors are thankful to professor Victor Ivanenko for fruitful discussions, remarks and suggestions.

\end{document}